\documentclass[submission,copyright,creativecommons]{eptcs}
\providecommand{\event}{AFL 2017} 

\usepackage{underscore}  
\usepackage{amsmath,amssymb,graphics}

\usepackage{tikz}
\usetikzlibrary{positioning,arrows.meta,decorations.pathreplacing,calc,fit}

\input tcilatex

\title{(Tissue) P Systems with Vesicles of Multisets}

\author{
Artiom Alhazov
\institute{
Institute of Mathematics and Computer Science\\
Academy of Sciences of Moldova\\
Academiei 5, Chi\c sin\u au, MD-2028, Moldova\\
\email{artiom@math.md}\\[0.4em]
}
\and 
Rudolf Freund
\institute{
Faculty of Informatics, TU Wien\\
Favoritenstra\ss{}e 9--11, 1040 Vienna, Austria\\
\email{rudi@emcc.at}\\[0.4em]
}
\and 
Sergiu Ivanov
\institute{
Laboratoire d'Algorithmique, Complexit\'e et Logique,\\
Universit\'{e} Paris Est -- Cr\'{e}teil Val de Marne\\
61, av. G\'{e}n\'{e}ral de Gaulle, 94010, Cr\'{e}teil, France\\
\email{sergiu.ivanov@u-pec.fr}\\[0.2em]
\and
TIMC-IMAG/DyCTiM, Faculty of Medicine of Grenoble,\\
5 avenue du Grand Sablon, 38700, La Tronche, France\\
\email{sergiu.ivanov@univ-grenoble-alpes.fr}\\[0.4em] 
}
\and
Sergey Verlan
\institute{
Laboratoire d'Algorithmique, Complexit\'e et Logique,\\
Universit\'e Paris Est Cr\'eteil,\\
61 av. du g\'en\'eral de Gaulle,  94010 Cr\'eteil, France\\
\email{verlan@u-pec.fr}
}
}

\def\titlerunning{(Tissue) P Systems with Vesicles of Multisets}
\def\authorrunning{A. Alhazov, R. Freund, S. Ivanov \& S. Verlan}

\begin{document}
\maketitle

\begin{abstract}
We consider tissue P systems working on vesicles of multisets with the very 
simple operations of insertion, deletion, and substitution of single objects. 
With the whole multiset being enclosed in a vesicle, sending it to a target cell 
can be indicated in those simple rules working on the multiset. As derivation 
modes we consider the sequential mode, where exactly one rule is applied 
in a derivation step, and the set maximal mode, where in each derivation step 
a non-extendable set of rules is applied. With the set maximal mode, 
computational completeness can already be obtained with tissue P systems 
having a tree structure, whereas tissue P systems even with an arbitrary 
communication structure are not computationally complete when working in 
the sequential mode. Adding polarizations -- -1, 0, 1 are sufficient -- allows 
for obtaining computational completeness even for tissue P systems working 
in the sequential mode. 
\end{abstract}


\section{Introduction}

Membrane systems were introduced at the end of last century by Gheorghe 
P\u{a}un, e.g., see \cite{DassowPaun1999} and \cite{Paun2000}, motivated 
by the biological interaction of molecules between cells and their surrounding 
environment. In the basic model, the membranes are organized in a 
hierarchical membrane structure (i.e., the connection structure between the 
compartments/regions within the membranes being representable as a tree),
and the multisets of objects in the membrane regions evolve in a maximally 
parallel way, with the resulting objects also being able to pass through the 
surrounding membrane to the parent membrane region or to enter an inner 
membrane. Since then, a lot of variants of membrane systems, for obvious
reasons mostly called \emph{P systems}, most of them being computationally 
complete, i.e., being able to simulate the computations of register machines. 
If an arbitrary graph is used as the connection structure between the 
cells/membranes, the systems are called \emph{tissue P systems}, see \cite{MartinVideetal2002}. 

Instead of multisets of plain symbols coming from a finite alphabet, 
P systems quite often operate on more complex objects (e.g., strings, 
arrays), too.  A comprehensive overview of different flavors of (tissue) P 
systems and their expressive power is given in the handbook which appeared 
in 2010, see~\cite{Paunetal2010}. For a state of the art snapshot of the domain, 
we refer the reader to the P systems website~\cite{Ppage}, as well as to the 
Bulletin series of the International Membrane Computing Society~\cite{imcs}.
\medskip

Very simple biologically motivated operations on strings are the so-called 
\emph{point mutations}, i.e., \emph{insertion}, \emph{deletion}, and 
\emph{substitution}, which mean inserting or deleting one symbol in a string 
or replacing one symbol by another one. For example, graph-controlled 
insertion-deletion systems have been investigated in~\cite{Freundetal2010},
and P systems using these operations at the left or right end of string objects 
were introduced in~\cite{Freundetal2014}, where also a short history of 
using these point mutations in formal language theory can be found.

When dealing with multisets of objects, the close relation of insertion and 
deletion with the increment and decrement instructions in register machines 
looks rather obvious. The power of changing states in connection with the 
increment and decrement instructions then has to be mimicked by moving 
the whole multiset representing the configuration of a register machine from
one cell to another one in the corresponding tissue system. Yet usually moving 
the whole multiset of objects in a cell to another one, besides maximal parallelism, 
requires \emph{target agreement} between all applied rules, i.e., that all results
are moved to the same target cell, e.g., see~\cite{FreundPaun2013}.
\medskip

In this paper we choose a different approach to guarantee that the 
whole multiset is moved even if only some point mutations are applied -- 
the multiset is enclosed in a vesicle, and this vesicle is moved from one cell 
to another one as a whole, no matter how many rules have been applied.
One constraint, of course, is that a common target has been selected by 
all rules to be applied; in the sequential derivation mode, this is no 
restriction at all, whereas in the set maximal derivation mode this means 
that the multiset of rules to be applied must be non-extendable, but all rules
must indicate the same target cell. As we will show, with the set maximal 
derivation mode computational completeness can be obtained, whereas with 
the sequential mode we achieve a characterization of the family of sets of 
(vectors of) natural numbers defined by partially blind register machines, 
which itself corresponds with the family of sets of (vectors of) natural numbers 
obtained as number (Parikh) sets of string languages generated by matrix 
grammars without appearance checking.
\smallskip

The idea of using vesicles of multisets has already been used in variants 
of P systems using the operations drip and mate, corresponding with the 
operations cut and paste well-known from the area of DNA computing, 
see~\cite{FreundOswald2007}. Yet in that case, always two vesicles (one 
of them possibly an axiom available in an unbounded number) have to interact. 
In this paper, the rules (bounded in number) are always applied to the same 
vesicle.

\medskip

The \emph{point mutations}, i.e., \emph{insertion}, \emph{deletion}, and 
\emph{substitution}, well-known from biology as operations on DNA, have 
also widely been used in the variants of \emph{networks of evolutionary 
processors (NEPs)}, which consist of cells (processors) each of them allowing 
for specific operations on strings. \emph{Networks of Evolutionary Processors} 
(NEPs) were introduced in~\cite{Castellanosetal2003} as a model of string 
processing devices distributed over a graph, with the processors carrying out 
these point mutations. Computations in such a network consist of alternatingly 
performing two steps -- an \emph{evolution step} where in each cell all possible 
operations on all strings currently present in the cell are performed, and a 
\emph{communication step} in which strings are sent from one cell to another 
cell provided specific conditions are fulfilled. Examples of such conditions are 
(output and input) filters which have to be passed, and these (output and input) 
filters can be specific types of regular languages or permitting and forbidden 
context conditions. The set of strings obtained as results of computations by the 
NEP is defined as the set of objects which appear in some distinguished node in 
the course of a computation. 

In \emph{hybrid networks of evolutionary processors} (HNEPs), each language 
processor performs only one of these operations at a certain position of the strings. 
Furthermore, the filters are defined by some variants of random-context conditions, 
i.e., they check the presence and the absence of certain symbols in the strings. 
For an overview on HNEPs and the best results known so far, we refer the reader 
to~\cite{Alhazovetal2016}.

In \emph{networks of evolutionary processors with polarizations}, each symbol 
has assigned a fixed integer value; the polarization of a string is computed 
according to a given evaluation function, and in the communication step the 
obtained string is moved to any of the connected cells having the same polarization.
Networks of polarized evolutionary processors were considered 
in~\cite{Arroyoetal2014} and~\cite{Arroyoetal2017}), and networks of evolutionary
processors only using the elementary polarizations $-1,0,1$ were investigated in \cite{Popescu2016}. The number of processors (cells) needed to obtain 
computational completeness has been improved in a considerable way  
in~\cite{Freundetal2017} making these results already comparable 
with those obtained in~\cite{Alhazovetal2016} for hybrid networks of evolutionary 
processors using permitting and forbidden contexts as filters for the communication 
of strings between cells. 

Seen from a biological point of view, networks of evolutionary processors are a 
collection of cells communicating via membrane channels, which makes them  
closely related to tissue-like P systems considered in the area of membrane computing.
Hence, in this paper we will also take over the idea of polarizations; as 
in~\cite{Popescu2016} and in~\cite{Freundetal2017}, we will only consider the 
elementary polarizations $-1,0,1$ for the symbols as well as for the cells.
Using this variant of tissue P systems, we are going to show computational 
completeness even with the sequential derivation mode.

The rest of the paper is structured as follows: In Section~\ref{Prerequisites} 
we recall some well-known definitions from formal language theory, and in 
the succeeding Section~\ref{tPsystems} we give the definitions of the model of 
tissue P systems with vesicles of multisets as well as its variants to be considered 
in this paper, especially the variant with elementary polarizations $-1,0,1$. 
In Section~\ref{Subst} we show our main results for tissue P systems with 
vesicles of multisets using all three operations insertion, deletion, and substitution, 
but without using polarizations, i.e., that computational completeness can be 
achieved by using the set maximally parallel derivation mode, whereas with the 
sequential mode we get a characterization of the families of sets of natural 
numbers and Parikh sets of natural numbers generated by partially blind 
register machines. In Section~\ref{Pol} we show that even with the sequential 
derivation mode we obtain computational completeness when using polarizations 
(only $-1,0,1$ are needed). A summary of the results and an outlook to future 
research conclude the paper.


\section{Prerequisites}\label{Prerequisites}

We start by recalling some basic notions of formal language theory. An
alphabet is a non-empty finite set. A finite sequence of symbols from an
alphabet $V$ is called a \emph{string} over $V$. The set of all strings over 
$V$ is denoted by $V^{\ast }$; the \emph{empty string} is denoted by $%
\lambda $; moreover, we define $V^{+}=V^{\ast }\setminus \{\lambda \}$. The 
\emph{length} of a string $x$ is denoted by $|x|$, and by $|x|_{a}$ we
denote the number of occurrences of a letter $a$ in a string $x$. For a
string $x$, $\emph{alph(x)}$ denotes the smallest alphabet $\Sigma $ such
that $x\in \Sigma ^{\ast }$. 

A \emph{multiset} $M$ with underlying set $A$ is a pair $(A,f)$ where 
$f:\,A\to \mathbb{N}$ is a mapping, with $\mathbb{N}$ denoting the set of 
natural numbers (non-negative integers). If $M=(A,f)$ is a multiset then its 
\emph{support} is defined as $supp(M)=\{x\in A \,|\,f(x)> 0\}$. A multiset is 
empty (respectively finite) if its support is the empty set (respectively a 
finite set). If $M=(A,f)$ is a finite multiset over $A$ and 
$supp(M)=\{ a_1,\ldots,a_k\}$, then it can also be represented by the string 
$a_1^{f(a_1)} \dots a_k^{f(a_k)}$ over the alphabet $\{ a_1,\ldots,a_k\}$
(the corresponding vector ${f(a_1)}, \dots a_,{f(a_k)}$ of natural numbers is 
called Parikh vector of the string $a_1^{f(a_1)} \dots a_k^{f(a_k)}$), and, 
moreover, all permutations of this string precisely identify the same multiset 
$M$ (they have the same Parikh vector). The set of all multisets over the 
alphabet $V$ is denoted by $V^{\circ}$.

The family of all recursively enumerable sets of strings is denoted by 
$RE$, the corresponding family of recursively enumerable sets of 
Parikh sets (vectors of natural numbers) is denoted by $PsRE$.
For more details of formal language theory the reader is referred to the 
monographs and handbooks in this area, such as 
\cite{RozenbergSalomaa1997}.


\subsection{Insertion, deletion, and substitution}
\label{PointMutations}

For an alphabet $V$, let $a\rightarrow b$ be a rewriting rule with $a,b\in
V\cup \{\lambda \}$, and $ab\neq \lambda $; we call such a rule a 
\emph{substitution rule} if both $a$ and $b$ are different from $\lambda $; 
such a rule is called a \emph{deletion rule} if $ a\neq \lambda $ and 
$b=\lambda $, and it is called an \emph{insertion rule} if $a=\lambda $ 
and $b\neq \lambda $. The set of all insertion rules, deletion rules, and 
substitution rules over an alphabet $V$ is denoted by $Ins_{V},Del_{V}$, 
and $Sub_{V}$, respectively. Whereas an insertion rule is always applicable, 
the applicability of a deletion and a substitution rules depends on the 
presence of the symbol $a$. We remark that insertion rules, deletion rules, 
and substitution rules can be applied to strings as well as to multisets, too.
Whereas in the string case, the position of the inserted, deleted, and 
substituted symbol matters, in the case of a multiset this only means 
incrementing the number of symbols $b$, decrementing the number of 
symbols $a$, or decrementing the number of symbols $a$ and at the 
same time incrementing the number of symbols $b$.


\subsection{Register machines}
\label{SubsectionRM}

Register machines are well-known universal devices for computing
(generating or accepting) sets of vectors of natural numbers.

\begin{definition}

A \emph{register machine} is a construct
\[
M=\left( m,B,l_{0},l_{h},P\right)
\]
where
\begin{itemize}
\item $m$ is the number of registers,

\item $B$ is a set of labels bijectively labeling the instructions in the set $P$,

\item $l_{0}\in B$ is the initial label, and

\item $l_{h}\in B$ is the final label.
\end{itemize}

The labeled instructions of $M$ in $P$ can be of the following forms:

\begin{itemize}
\item $p:\left( ADD\left( r\right) ,q,s\right) $, with $p\in
B\setminus \left\{ l_{h}\right\} $, $q,s\in B$, $1\leq r\leq m$.%
\newline
Increase the value of register $r$ by one, and non-deterministically jump to
instruction $q$ or $s$.

\item $p:\left( SUB\left( r\right) ,q,s\right) $, with $p\in
B\setminus \left\{ l_{h}\right\} $, $q,s\in B$, $1\leq r\leq m$.%
\newline
If the value of register $r$ is not zero then decrease the value of 
register~$r$ by one (\emph{decrement} case) and jump to instruction 
$q$, otherwise jump to instruction~$s$
(\emph{zero-test} case).

\item $l_{h}:HALT$.\newline
Stop the execution of the register machine.
\end{itemize}

A \emph{configuration} of a register machine is described by the contents
of each register and by the value of the current label, which indicates the
next instruction to be executed.

\end{definition}
\smallskip

In the accepting case, a computation starts with the input of a $k$-vector
of natural numbers in its first $k$ registers and by executing the first
instruction of $P$ (labeled with $l_{0}$); it terminates with reaching the
$HALT$-instruction. Without loss of generality, we may assume all
registers to be empty at the end of the computation.

In the generating case, a computation starts with all registers
being empty and by executing the first instruction of $P$ (labeled
with $l_{0}$); it terminates with reaching the $HALT$-instruction
and the output of a $k$-vector of natural numbers in its first $k$ registers.
Without loss of generality, we may assume all registers $>k$ to be empty
at the end of the computation. The set of vectors of natural numbers 
computed by $M$ in this way is denoted by $Ps(M)$. If we want to generate 
only numbers ($1$-dimensional vectors), then we have the result of a 
computation in register $1$ and the set of numbers computed by $M$ in this 
way is denoted by $N(R)$. By $NRM$ and $PsRM$ we denote the families of 
sets of natural numbers and of sets of vectors of natural numbers, respectively, 
generated by register machines. It is folklore (e.g., see~\cite{Minsky1967}) that 
$PsRE = PsRM$ and $NRE = NRM$ (actually, three registers are sufficient in 
order to generate any set from the family $NRE$, and, in general, $k+2$ 
registers needed to generate any set of from the family $NRE$).

\subsubsection{Partially blind register machines}

In the case when a register machine cannot check whether a register is empty 
we say that it is partially blind: the registers are increased and decreased by 
one as usual, but if the machine tries to subtract from an empty register, then 
the computation aborts without producing any result (that is we may say that the 
subtract instructions are of the form $p:\left( SUB\left( r\right) ,q,abort\right) $; 
instead, we simply will write $p:\left( SUB\left( r\right) ,q,abort\right) $. 
Moreover, acceptance or generation now by definition also requires 
all registers, except the first $k$ output registers, to be empty (which means all 
registers $k + 1, . . . , m$ have to be empty at the end of the computation), 
i.e., there is an implicit test for zero, at the end of a (successful) computation,  
that is why we say that the device is partially blind. By $NPBRM$ and $PsPBRM$ 
we denote the families of sets of natural numbers and of sets of vectors of natural 
numbers, respectively, computed by partially blind register machines. It is known 
(e.g., see~\cite{Freundetal2005}) that partially blind register machines are strictly 
less powerful than general register machines (hence than Turing machines); 
moreover, $NPBRM$ and $PsPBRM$ characterize the number and Parikh sets,  
respectively, obtained by matrix grammars without appearance checking.


\section{Tissue P systems working on vesicles of multisets}
\label{tPsystems}

We first define our basic model of tissue P systems working on vesicles of 
multisets in the maximally parallel set derivation mode:

\begin{definition}
A \emph{tissue P systems working on vesicles of multisets} (a \emph{tPV} 
system for short) is a tuple 
\[
\Pi = \left( L,V,T,R,(i_0,w_0),h\right) 
\]
where
\begin{itemize}
\item $L$ is a set of labels identifying in a one-to-one manner the $|L|$ cells of 
the tissue P system $\Pi$;

\item $V$ is the alphabet of the system,

\item $T$ is the terminal alphabet of the system,

\item $R$ is a set of rules of the form $(i,p,j)$ where 
$i,j\in L$ and $p\in Ins_{V}\cup Del_{V}\cup Sub_{V}$, i.e., 
$p$ is an insertion, deletion or substitution rule over the alphabet $V$;
we may collect all rules from cell $i$ in one set and then write 
$R_i=\{(i,p,j)\mid (i,p,j)\in R\}$, so that $R=\bigcup_{i\in L}R_i$;
moreover, for the sake of conciseness, we may simply write 
$R_i=\{(p,j)\mid (i,p,j)\in R\}$, too;

\item $(i_0,w_0)$ describes the initial vesicle containing the multiset 
$w_0$ in cell $i_0$.
\end{itemize}
\end{definition}

As in the case of NEPs and HNEPs, we call $\Pi $ a \emph{hybrid} tPV 
system if every cell is ``specialized'' in one type of evolution rules from 
(at most) one of the sets $Ins_{V},Del_{V}$, and $Sub_{V}$, respectively.
\bigskip

The tPV system can work with different derivation modes for applying the 
rules in $R$. The simplest case is the sequential mode (abbreviated 
\emph{sequ}), where in each derivation step, with the vesicle enclosing the
multiset $w$ being in cell $i$, exactly one rule $(i,p,j)$ from $R_i$ is applied, 
which in fact means that $p$ is applied to $w$ and the resulting multiset in its 
vesicle is moved to cell $j$. Using the set maximally parallel derivation mode
(abbreviated \emph{smax}), with the vesicle enclosing the multiset $w$ being 
in cell $i$, we apply a non-extendable multiset of rules from $R_i$, which has 
to obey the condition that all the evolution rules $(i,p,j)$ in this multiset of 
rules specify the same target cell $j$.

In any case, the computation of $\Pi $ starts with a vesicle containing the 
multiset $w_0$ in cell $i_0$, and the computation proceeds in the underlying 
derivation mode until an output condition is fulfilled, which in all possible cases 
means that the vesicle has arrived in the output cell $h$. As we are dealing 
with membrane systems, the classic additional condition may be that the 
computation halts, i.e., in cell $h$ no rule can be applied any more to the 
multiset in the vesicle which has arrived there. As we have also specified 
a terminal alphabet, another condition -- for its own or in combination with 
halting -- is that the multiset in the vesicle which has arrived in cell $h$ only 
contains terminal symbols. Hence, we may specify one of the following 
output strategies:
\begin{itemize}
\item $halt$: the only condition is that the system halts, the result is the 
multiset contained in the vesicle to be found in cell~$h$ (which in fact means 
that specifying the terminal alphabet is obsolete);

\item $term$: the resulting multiset contained in the vesicle to be found in
cell~$h$ consists of terminal symbols only (yet the system need not have 
reached a halting configuration).

\item $(halt,term)$: both conditions must be fulfilled, i.e., the system halts 
and the resulting multiset contained in the vesicle to be found in cell~$h$ 
consists of terminal symbols only.
\end{itemize}

The set of all multisets obtained as results of computations in $\Pi$ 
working in the derivation mode $\alpha \in \{sequ, smax\}$ with the 
output being obtained by taking the output condition 
$\beta \in \{halt,term,(halt,term)\}$ is denoted by 
$Ps(\Pi,\alpha ,\beta )$; if we are only interested in the number of symbols 
in the resulting multiset, the corresponding set of natural numbers is denoted 
by $N(\Pi,\alpha ,\beta )$. The families of sets of ($k$-dimensional) vectors 
of natural numbers and sets of natural numbers generated by tPV systems 
with at most $n$ cells working in the derivation mode $\alpha$ and using the 
output strategy $\beta $ are denoted by $Ps(tPV_n,\alpha ,\beta )$ 
($Ps^{k}(tPV_n,\alpha ,\beta )$) and $N(tPV_n,\alpha ,\beta )$, respectively. 
If $n$ is not bounded, we simply omit the subscript in these notations.

\bigskip

We should like to mention that the communication structure between the 
cells in a tPV system is implicitly given by the rules in $R$, i.e., the 
underlying (directed! graph) $G=(N,E)$ with $N$ being the set of nodes
and $E$ being the set of (directed) edges is given by 
\begin{itemize}
\item $N=L$ and

\item $E=\left\{ (i,j) \mid (i,p,j) \in R\right\}$.
\end{itemize}

In general, we do not forbid $G$ to have loops. Moreover, if $G$ can be 
interpreted as a tree, then we call the tPV system a \emph{{\textbf{hierarchical}} 
P system working on vesicles of multisets} (abbreviated \emph{PV system}); in all 
definitions given above for the families of sets of (vectors of) natural 
numbers we then write $PV$ instead of $tPV$.


\section{Results for tissue P systems with vesicles of multisets}
\label{Subst}

Our first result shows that with the derivation mode $smax$ and using all three
types of point mutation rules computational completeness can even be 
obtained with PV systems:

\begin{theorem}\label{TheoremPVsmax}
$PsRE \subseteq Ps(PV,smax ,\beta )$ for any
$\beta \in \{(halt,term),halt,term\}$.
\end{theorem}

\begin{proof}
Let $K$ be an arbitrary recursively enumerable set of $k$-dimensional 
vectors of natural numbers. Then $K$ can be generated by a register 
machine $M$ with two working registers also using decrement instructions 
and $k$ output registers. In order to have a general construction, we do not 
restrict the number of working registers in the following. Let 
$M=\left( m,B,l_{0},l_{h},P\right) $ be a register machine generating $K$.
\medskip

We now define a PV system $\Pi $ generating $K$, 
i.e., $Ps(\Pi ,smax ,\beta )=K$:
\begin{align*}
\Pi		&=\left( L,V,T,R,(i_0,w_0),h\right) ,\\
L		&=\{ r\mid 1\leq r\leq k\} \cup 
		     \{ r,r_-,r_0\mid k+1\leq r\leq m\} \cup \{ h\} ,\\
V		&=  L\cup \{ a_r\mid 1\leq r\leq m\} \cup \{ \# \},\\
T		&=  \{ a_r\mid 1\leq r\leq k\},\\
R		&=  \{ (0,p\to q,r),(0,p\to s,r),(r,\lambda \to a_r,0)\mid 
			p:\left( ADD\left( r\right) ,q,s\right) \in P\},\\
		&\cup \{ (0,p\to q,r_- ),(0,p\to s,r_0 )\mid 
			p:\left( SUB\left( r\right) ,q,s\right) \in P\}\\
		&\cup \{ (r_- , a_r \to \lambda ,0), (r_0 , s \to s,0), (r_0 , a_r \to \# ,0)\mid 
			p:\left( SUB\left( r\right) ,q,s\right) \in P\},\\
		&\cup \{ (0,l_{h}\to \lambda ,h),
			      (h,\# \to \# ,0),(0,\# \to \# ,h)\} ,\\
(i_0,w_0)	&=(0,l_{0}).
\end{align*}

\begin{figure}[htb]
\begin{center}
  %

  \pgfkeys{
    /pgf/arrow keys/.cd,
    shear/.store in=\pgfarrowshear,
    DAhead/.style={
      length = +2pt 1.925 1,
      shear = .8mm,
    },
  }

  \makeatletter
  \newdimen\pgfarrowshear
  \let\oldmacro\pgf@arrow@drawer@shift
  \def\pgf@arrow@drawer@shift{\pgftransformyshift\pgfarrowshear\oldmacro}

  \tikzstyle doublearr=[double,double distance=1.5mm,{<[DAhead]}-{>[DAhead]},
    shorten <=-.1mm,shorten >=-.1mm]

  \begin{tikzpicture}[node distance=5mm and 20mm]
    \tikzstyle label=[draw,circle,minimum size=8mm,inner sep=1mm]
    \tikzstyle var=[densely dashed]

    \node[label] (0) {0};
    \node[label,var,below left=of 0] (r) {$r$};
    \node[label,below right=of 0] (h) {$h$};

    \newcommand{\suboff}{8mm}
    \node[label,var,below right=of r,xshift=-\suboff,inner sep=.5mm] (r-) {$r_-$\strut};
    \node[label,var,below left=of h,xshift=\suboff,inner sep=.5mm] (r0) {$r_0$\strut};

    \draw[doublearr] (0) -- (r);
    \draw[doublearr] (0) -- (h);
    \draw[doublearr] (0) -- (r-);
    \draw[doublearr] (0) -- (r0);

    \node[below=0mm of r] {$ADD(r)$};

    \draw[decorate,decoration={brace,amplitude=1.5mm}]
    ($(r0.south east)+(0,-2mm)$) -- ($(r-.south west)+(0,-2mm)$);
    \node at ($(r-)!.5!(r0)+(1mm,-10mm)$) {$SUB(r)$};

    \node[below=0mm of h,align=left,xshift=2.5mm] {halting\\and trap};
  \end{tikzpicture}
\end{center}
\vspace{-3mm}

\caption{Communication structure of the two-level hierarchical PV
  system.  Each node with a dashed contour is replicated for every
  register $r$.}
\end{figure}
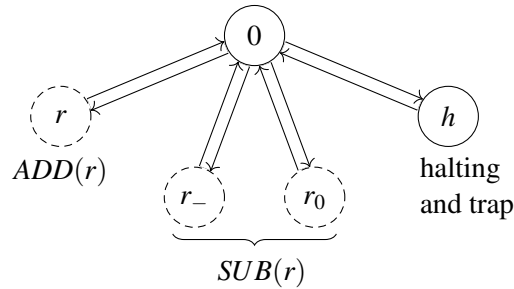

The root of the communication tree is cell $0$. From there, all simulations of 
register machine instructions are initiated:
\begin{description}
\item[$\left( ADD\left( r\right) ,q,s\right) $] is simulated by moving the vesicle from
the root cell to cell $r$ by applying one of the rules from 
$\{ (0,p\to q,r),(0,p\to s,r),(r,\lambda \to a_r,0) \}$; in cell $r$ the number of 
symbols $a_r$ representing the contents of register $r$ is incremented by 
the insertion rule $(r,\lambda \to a_r,0) $, which also sends back the 
vesicle to the root cell.

\item[$\left( SUB\left( r\right) ,q,s\right) $] is simulated by first choosing one of 
the rules from $\{ (0,p\to s,r_0 ),(0,p\to q,r_- )\} $ in a non-deterministic way, 
guessing whether the number of symbols $a_r$ representing the contents of 
register $r$ is zero or not. If the number is not zero, then in cell $r_-$ the 
deletion operation in the rule $(r_- , a_r \to \lambda ,0)$ can be carried out 
and the vesicle is sent back to cell $0$, whereas otherwise the vesicle gets 
stuck in cell $r_-$ and therefore no result can be obtained in the output 
cell~$h$. If the number of symbols $a_r$ has been assumed to be zero and 
the vesicle is in cell $r_0$, then there the rule $(r_0 , s \to s,0)$ can be applied 
in any case, and the vesicle is sent back to cell $0$. Yet if the assumption has 
been wrong, then in parallel the rule $(r_0 , a_r \to \# ,0)$ must be applied, 
thus introducing the trap symbol $\# $. This is the only case in the whole 
construction where the possibility of applying (at least) two rules in parallel 
is used for appearance checking. We point out that both rules have the same 
target $0$.
\end{description}

Any halting computation in $M$ finally reaches the halting instruction labeled 
by $l_h$, and thus in $\Pi $, by applying the rule $(0,l_{h}\to \lambda ,h)$, 
the vesicle obtained so far is moved to the final cell $h$. Provided no trap 
symbol $\#$ has been generated during the simulation of the computation 
in $M$ by the tPV system $\Pi $, the multiset in this vesicle only contains
terminal symbols and the computation in $\Pi $ halts as well.

In sum, we conclude that $Ps(\Pi ,smax ,\beta )=K$ for any
$\beta \in \{(halt,term),halt,term\}$.
\hfill
\end{proof}
\bigskip

The construction given in the preceding proof offers some additional nice 
features:
\begin{itemize}
\item The PV system $\Pi $ is a hybrid one, as in each cell only one kind of 
rules is employed: substitution in cells $0$ and $h$ and in cells $r_0$, 
insertion in cells $r$, deletion in cells $r_-$.

\item The trap rules $(h,\# \to \# ,0),(0,\# \to \# ,h)$, guaranteeing a 
non-halting computation as soon as the introduction of the trap symbol $\# $ 
has been enforced by a wrong guess, are only needed in the case of 
the output strategy $halt$.

\item The vesicle must always leave the current cell whenever a rule 
can be applied.

\item The number of cells in the PV system $\Pi $ only depends on the 
number of registers in the register machine $M$. Suppose $M$ has 
$k$ output registers and $2$ working registers. Since the output 
registers are never decremented, we only need one cell $r$ for each 
such register.  We need 3 cells ($r$, $r_-$, and~$r_0$) for each of
the two working (decrementable) registers.  Finally, we need the 
cells $0$ and $h$, which amounts in a total of $k+2\cdot 3+2=k+8$ 
cells to simulate $M$.  This also means that only 9 cells are 
needed for generating number sets. 
\end{itemize}

If the underlying register machine is partially blind, we only have to 
consider the decrement case, which then still works correctly, whereas 
we can omit the zero test case, and thus can omit the parallelism.
Hence, we immediately infer the following result:

\begin{theorem}\label{TheoremPBRM}
$PsPBRM \subseteq Ps(PV,sequ ,\beta )$ for any
$\beta \in \{(halt,term),halt,term\}$.
\end{theorem}

\begin{proof}
Let $K\in PsPBRM$, i.e., the vector set $K$ can be generated by a partially 
blind register machine $M=\left( m,B,l_{0},l_{h},P\right) $.
As in the preceding proof, we now define a PV system $\Pi $ generating $K$ 
in the sequential derivation mode, i.e., $Ps(\Pi ,sequ,\beta )=K$:
\begin{align*}
\Pi		&=\left( L,V,T,R,(i_0,w_0),h\right) ,\\
L		&=\{ r\mid 1\leq r\leq k\} \cup 
		     \{ r,r_-\mid k+1\leq r\leq m\} \cup \{ h\} ,\\
V		&=  L\cup \{ a_r\mid 1\leq r\leq m\} \cup \{ \# \},\\
T		&=  \{ a_r\mid 1\leq r\leq k\},\\
R		&=  \{ (0,p\to q,r),(0,p\to s,r),(r,\lambda \to a_r,0)\mid 
			p:\left( ADD\left( r\right) ,q,s\right) \in P\},\\
		&\cup \{ (0,p\to q,r_- ), (r_- , a_r \to \lambda ,0)\mid 
			p:\left( SUB\left( r\right) ,q\right) \in P\},\\
		&\cup \{ (0,l_{h}\to \lambda ,h),
			      (h,\# \to \# ,0),(0,\# \to \# ,h)\} 
		    \cup \{ (h,a_r \to \# ,0)\mid k+1\leq r\leq m\},\\
(i_0,w_0)	&=(0,l_{0}).
\end{align*}

The simulation of the computations in $M$ by $\Pi $ works in a similar way
as in the preceding proof, with the main reduction that no zero test case 
has to be simulated, hence, everything can be carried out in a sequential way.

Any halting computation in $M$ finally reaches the halting instruction labeled 
by $l_h$, and thus in $\Pi $, by applying the rule $(0,l_{h}\to \lambda ,h)$, 
the vesicle obtained so far is moved to the final cell $h$. Provided no 
non-terminal symbol $a_r$ with $ k+1\leq r\leq m$ is still present, the 
computation in $\Pi $ will halt, but otherwise the trap symbol $\#$ will be 
introduced by (one of) the rules from $ \{ (h,a_r \to \# ,0)\mid k+1\leq r\leq m\}$,
thus causing an infinite loop.

In sum, we conclude that $Ps(\Pi ,sequ ,\beta )=K$ for any
$\beta \in \{(halt,term),halt,term\}$.
\hfill
\end{proof}

\bigskip

The following corollary is immediate consequence of 
Theorem~\ref{TheoremPVsmax} proved above:

\begin{corollary}\label{CorollarytPV}
$PsRE = Ps(PV,smax ,\beta ) = Ps(tPV,smax ,\beta )$ for any
$\beta \in \{(halt,term),halt,term\}$.
\end{corollary}
\begin{proof} 
By definition, any PV system is a tPV system, too. Hence, it only remains 
to show that $Ps(tPV,smax ,\beta ) \subseteq PsRE$, yet we  omit a direct 
construction as the result can be inferred from the Turing-Church thesis.
\hfill
\end{proof}

\bigskip

We now also show that the computations of a sequential tPV system 
using the output strategy $term$ can be simulated by a partially blind 
register machine. 

\begin{theorem}\label{TheoremPBRMrev}
$Ps(tPV,sequ ,term) \subseteq PsPBRM$.
\end{theorem}

\begin{proof} (\textit{Sketch})
Let $\Pi=\left( L,V,T,R,(i_0,w_0),h\right) $ be an arbitrary tPV system working 
in the sequential derivation mode yielding an output in the output cell provided 
the multiset in the vesicle having arrived there contains only terminal symbols;
without loss of generality we assume $L=\{ i\mid 1\leq i\leq l\}$.

We now construct a register machine $M=\left( m,B,l_{0},l_{h},P\right) $ 
generating $Ps(\Pi ,sequ ,term )$, yet using a more relaxed definition for 
the labeling of instructions in $M$, i.e., one label may be used for different 
instructions, which does not affect the computational power of the register 
machine as shown in~\cite{Freundetal2005}. For example, instead of a 
nondeterministic ADD-instruction $p:\left( ADD\left( r\right) ,q,s\right)$ we
use the two ADD-instructions $p:\left( ADD\left( r\right) ,q\right)$ and 
$p:\left( ADD\left( r\right) ,s\right)$. Moreover, we omit the generation of 
$w_0$ in $l_0$ by a sequence of $ADD$-instructions finally ending up with 
label $l_0$ and the correct values in registers $r$ for the numbers of symbols 
$a_r$ in cell $l_0$. 
\medskip

We now sketch how the rules of $\Pi $ can be simulated by register machine 
instructions in $M$:

\begin{description}
\item[$(i,\lambda \to b,j)$ ] is simulated by 
	$i:\left( ADD\left( b\right) ,j\right) $.

\item[$(i,a\to \lambda ,j)$ ] is simulated by 
	$i:\left( SUB\left( a\right) ,j\right) $.

\item[$(i,a\to b,j)$ ] is simulated by the sequence of two instructions
	$i:\left( SUB\left( a\right) ,i^{\prime }\right) $ and 
	$i^{\prime }:\left( ADD\left( b\right) ,j\right) $ using an intermediate 
	label $i^{\prime }$.
\end{description}
Hence, for these simulations we may need $2l$ labels in the sense explained 
above. If a vesicle reaches the final cell $h$ with the multiset inside only 
consisting of terminal symbols, we also have to allow $M$ to have this 
multiset as a result: this goal can be accomplished by using the final 
sequence 
\begin{align*}
 &h:\left( ADD\left( 1\right) ,\tilde{h}\right) ,\\
 &\tilde{h}:\left( SUB\left( 1\right) ,\hat{h}\right) ,\\
 &\hat{h}:HALT.
\end{align*}
We observe that $\tilde{h},\hat{h}$ are labels different from $h^{\prime }$.
Since $\hat{h}$ is now the only halting instruction of $M$, it must
reset to zero all its working registers before reaching $\hat{h}$ to
satisfy the final zero check, which corresponds to $\Pi$ producing a
multiset consisting exclusively of terminal symbols.

In sum, we conclude that $Ps(M)=Ps(\Pi ,sequ ,term)$.
\hfill
\end{proof}
\bigskip

As a consequence of Theorems~\ref{TheoremPBRM} 
and~\ref{TheoremPBRMrev} we obtain:

\begin{corollary}\label{CorollaryPBRM}
$PsPBRM = Ps(PV,sequ ,term)$.
\end{corollary}


\section{Polarized tissue P systems with vesicles of multisets}
\label{Pol}

In a polarized tissue P system $\Pi $ working on vesicles of multisets, each cell 
gets assigned an elementary polarization from $\{-1,0,1\}$; each symbol from the 
alphabet $V$ also has an integer polarization but every terminal symbol from the 
terminal alphabet has polarization $0$. As we shall see later, we can even restrict 
ourselves to elementary polarizations from $\{-1,0,1\}$ for each symbol, too.

Given a multiset, we need an evaluation function computing the polarization of 
the whole multiset from the polarizations of the symbols it contains. Given the 
result $m$ of this evaluation of the multiset in the vesicle, we apply the
sign function $sign(m)$, which returns one of the values $+1/0/-1$, provided 
that $m$ is a positive integer / is $0$ / is a negative integer, respectively. 

The main difference between polarized tPV systems and normal tPV systems,
besides the polarizations assigned to symbols and multisets as well as to the 
cells, is the way the resulting vesicles are moved from one cell to another one:
although in the rules themselves still a target is specified, the vesicle can only 
move to a cell having the same polarization as the multiset contained in it. 
As a special additional feature we require that the vesicle must not stay 
in the current cell even if its polarization would fit (if there is no other cell 
with a fitting polarization, the vesicle is eliminated from the system).
As by the convention mentioned above we assume every terminal symbol 
from the terminal alphabet to have polarization $0$, it is necessary that the 
output cell itself also has to have polarization $0$.

\begin{definition}
A \emph{polarized tissue P systems working on vesicles of multisets} 
(a \emph{ptPV} 
system for short) is a tuple 
\[
\Pi = \left( L,V,T,R,(i_0,w_0),h,\pi _L,\pi _V,\varphi \right) 
\]
where
\begin{itemize}
\item $L$ is a set of labels identifying in a one-to-one manner the $|L|$ 
cells of the tissue P system $\Pi$;

\item $V$ is the polarized alphabet of the system,

\item $T$ is the terminal alphabet of the system (the terminal symbols 
have no polarization, i.e., polarization $0$),

\item $R$ is a set of rules of the form $(i,p,j)$ where 
$i,j\in L$ and $p\in Ins_{V}\cup Del_{V}\cup Sub_{V}$, i.e., 
$p$ is an insertion, deletion or substitution rule over the alphabet $V$;
we may collect all rules from cell $i$ in one set and then write 
$R_i=\{(i,p,j)\mid (i,p,j)\in R\}$, so that $R=\bigcup_{i\in L}R_i$;
moreover, for the sake of conciseness, we may simply write 
$R_i=\{(p,j)\mid (i,p,j)\in R\}$, too;

\item $(i_0,w_0)$ describes the initial vesicle containing the multiset 
$w_0$ in cell $i_0$;

\item $\pi _L$ is the function assigning an integer polarization to each 
cell (as already mentioned above, we here restrict ourselves to the 
elementary polarizations from $\{-1,0,1\}$);

\item $\pi _V$ is the function assigning an integer polarization to each 
symbol in $V$ (as already mentioned above, we here restrict ourselves 
to the elementary polarizations from $\{-1,0,1\}$);

\item $\varphi $ is the evaluation function yielding an integer value for each 
multiset.
\end{itemize}
\end{definition}

As in the case of NEPs and HNEPs, we call $\Pi $ a \emph{hybrid} ptPV 
system if a cell is ``specialized'' in one type of evolution rules from  
(at most) one of the sets $Ins_{V},Del_{V}$, and $Sub_{V}$, respectively.
\bigskip

The ptPV system again can work with different derivation modes for applying 
the rules in $R$, e.g., the sequential mode \emph{sequ} or the set maximally 
parallel derivation mode \emph{smax}. Yet a derivation step now consists 
of two substeps -- the \emph{evolutionary step} with applying the rule(s) from 
$R$ in the way required by the derivation mode (caution: we allow the set 
of applied rules to be empty) and the \emph{communication step} with sending 
the vesicle to a cell with the same polarization as the multiset in it.

In the following, we will only use the evaluation function $\varphi$ which 
computes the value of a multiset as the sum of the values 
of the symbols contained in it; we write $\varphi _s$ for this function.

In any case, the computation of $\Pi $ starts with a vesicle containing the 
multiset $w_0$ in cell $i_0$ (obviously, the initial multiset $w_0$ has to have 
the same polarization as the initial cell $i_0$), and the computation proceeds 
using the underlying derivation mode for the evolutionary steps until an 
output condition is fulfilled, which in all possible cases means that the vesicle 
has arrived in the output cell $h$. Again we use one of the output strategies
$halt$, $term$ and $(halt,term)$.

The set of all multisets obtained as results of computations in $\Pi$ 
working in the derivation mode $\alpha \in \{sequ, smax\}$, using the 
evaluation function $\varphi _s$ and the output condition 
$\beta \in \{halt,term,(halt,term)\}$, is denoted by $Ps(\Pi,\alpha ,\beta )$; 
if we are only interested in the number of symbols in the resulting multiset, 
the corresponding set of natural numbers is denoted by $N(\Pi,\alpha ,\beta )$. 
The families of sets of ($k$-dimensional) vectors of natural numbers and sets 
of natural numbers generated by ptPV systems with at most $n$ cells working 
in the derivation mode $\alpha$ and using the output strategy $\beta $ are 
denoted by $Ps(ptPV_n,\alpha ,\beta )$ ($Ps^{k}(ptPV_n,\alpha ,\beta )$) and 
$N(ptPV_n,\alpha ,\beta )$, respectively. If $n$ is not bounded, we simply omit 
the subscript in these notations.

\bigskip

We should like to mention that again the communication structure between 
the cells in a ptPV system is implicitly given by the rules in $R$, i.e., the 
underlying (directed! graph) $G=(N,E)$ with $N$ being the set of nodes
and $E$ being the set of (directed) edges is given by 
\begin{itemize}
\item $N=L$ and

\item $E=\left\{ (i,j) \mid (i,p,j) \in R\right\}$.
\end{itemize}

In general, we do not forbid $G$ to have loops. Moreover, if $G$ can be 
interpreted as a tree, then we call the ptPV system $\Pi $ a 
\emph{{\textbf{hierarchical}} polarized P system working on vesicles of multisets} 
(abbreviated \emph{pPV system}); in all definitions given above for the families of 
sets of (vectors of) natural numbers we then write $pPV$ instead of $ptPV$.
\bigskip

Moreover, there is another variant of interpreting the functioning of the 
ptPV $\Pi $ if $G$ is interpreted as an undirected graph 
$(L,\left\{ \left\{i,j\right\} \mid (i,p,j) \in R\right\})$. Then we may adopt 
the way of communication from polarized HNEPs and instead of 
specifying the set of rules as given above, change the definition in the 
following way:
\[
\Pi = \left( L,V,T,R,(i_0,w_0),h,\pi _L,\pi _V,\varphi ,G\right) 
\]
where $G$ now is an undirected graph defining the communication structure 
between the cells, and the rules in $R$ are specified without targets, i.e., 
they are written as $(i,p)$ instead of $(i,p,j)$ as the targets now are specified 
by the communication graph $G$. Yet as $G$ is an \emph{undirected} graph 
this makes a big difference as communication now by default is bidirectional,
i.e., we cannot enforce the direction of the movement of the vesicle any more.
According to these explanations it becomes obvious that this variant is a 
special case of ptPV systems. In fact, in this variant, if $(i,p,j)$ 
is a rule in $R$, then also $(j,p,i)$ must be a rule in $R$. As a special 
variant of ptPV systems, we then call it a uptPV system (with u specifying 
that the communication structure is an undirected graph).

Even with uptPV systems we can obtain computational completeness with 
the sequential derivation mode:
\begin{theorem}
$PsRE\subseteq Ps(uptPV_n,sequ ,term )$.
\end{theorem}

\begin{proof}
Let $M=\left( m,B,l_{0},l_{h},P\right) $ be an arbitrary register machine 
generating $k$-dimensional vectors. We now construct a uptPV system 
$\Pi $ generating the same set of multisets as $M$, i.e.,
$Ps(\Pi ,sequ,term )=Ps(M)$.
\medskip

\begin{figure}[htb]
  \begin{center}
    \begin{tikzpicture}[node distance=5mm and 12mm]
      \tikzstyle cell=[draw,circle,minimum size=8mm,inner sep=.5mm]
      \tikzstyle var=[densely dashed]
      \tikzstyle zone=[draw,densely dotted]
      \newcommand{\pol}[1]{\small $\langle #1 \rangle$}

      \node[cell,
        label={[inner sep=0mm]90:\pol{0}}
      ] (0) {0};

      \node[cell,above left=of 0,
        label={[inner sep=0mm,name=0' lab]165:\pol{0}}
      ] (0') {$0'$};
      \node[cell,var,left=of 0,
        label={[inner sep=0mm,name=r+ lab]165:\pol{+}}
      ] (r+) {$r_+$\strut\hspace{-.6mm}~};
      \node[cell,var,below left=of 0,
        label={[inner sep=0mm]165:\pol{+}}
      ] (tilde r+) {$\tilde r_+$\strut};

      \newcommand{\topang}{35}
      \newcommand{\botang}{-35}

      \node[cell,var,below right=of 0,
        label={[inner sep=0mm]\botang:\pol{+}}
      ] (r-) {$r_-$\strut};
      \node[cell,var,above right=of 0,
        label={[inner sep=0mm]\topang:\pol{-}}
      ] (r0) {$r_0$\strut};

      \tikzset{node distance=5mm and 7mm}
      \node[cell,var,right=of r-,
        label={[inner sep=0mm]\botang:\pol{0}}
      ] (tilde r-) {$\tilde r_-\strut$};
      \node[cell,var,right=of tilde r-,
        label={[inner sep=0mm]\botang:\pol{-}}
      ] (bar r-) {$\bar r_-\strut$};
      \node[cell,right=of bar r-,
        label={[inner sep=0mm,name=0- lab]\botang:\pol{0}}
      ] (0-) {$0_-$\strut\hspace{-.6mm}~};

      \node[cell,var,right=of r0,
        label={[inner sep=0mm]\topang:\pol{-}}
      ] (tilde r0) {$\tilde r_0$\strut};
      \node[cell,draw=none,right=of tilde r0] (dummy) {\phantom{$\tilde r_0$}\strut};
      \node[cell,right=of dummy,
        label={[inner sep=0mm,name=00 lab]\topang:\pol{0}}
      ] (00) {$0_0$\strut\hspace{-.6mm}~};

      \draw (0) -- (0');
      \draw (0') -- (r+);
      \draw (r+) -- (tilde r+) -- (0);

      \draw (0) -- (r0);
      \draw (r0) -- (tilde r0);
      \draw (tilde r0) -- (00);

      \draw (0) -- (r-);
      \draw (r-) -- (tilde r-);
      \draw (tilde r-) -- (bar r-);
      \draw (bar r-) -- (0-);

      \draw (00) to[out=-165,in=7] (0);
      \draw (0-) to[out=155,in=-2] (0);

      \node[zone,fit={(0' lab) (r+ lab) (tilde r+)}] (add zone) {};
      \node[below=0 of add zone] {$ADD(r)$};

      \node[zone,fit={(r0) (00 lab)}] (sub0 zone) {};
      \node[above=0 of sub0 zone] {$SUB(r)$, empty register $r$};

      \node[zone,fit={(r-) (0- lab)}] (sub- zone) {};
      \node[below=0 of sub- zone] {$SUB(r)$, successful decrement of $r$};

      \newcommand{\haltang}{165}
      \node[cell,below=of 0,yshift=-5mm,
        label={[inner sep=0mm,name=lh lab]\haltang:\pol{0}}
      ] (lh) {$l_h$};
      \node[cell,below=of lh,
        label={[inner sep=0mm]\haltang:\pol{0}}
      ] (tilde lh) {$\tilde l_h$};
      \draw (0) -- (lh) -- (tilde lh);
      \node[zone,fit={(lh lab) (tilde lh)}] (halt zone) {};
      \node[below=0 of halt zone,xshift=.5mm] {halting};
    \end{tikzpicture}
  \end{center}

  \caption{The communication graph $G$ of the computationally complete
    uptPV system.  We also represent the polarizations of the nodes in
    angular brackets.  Each node with a dashed contour is replicated
    for every register $r$.}
\end{figure}
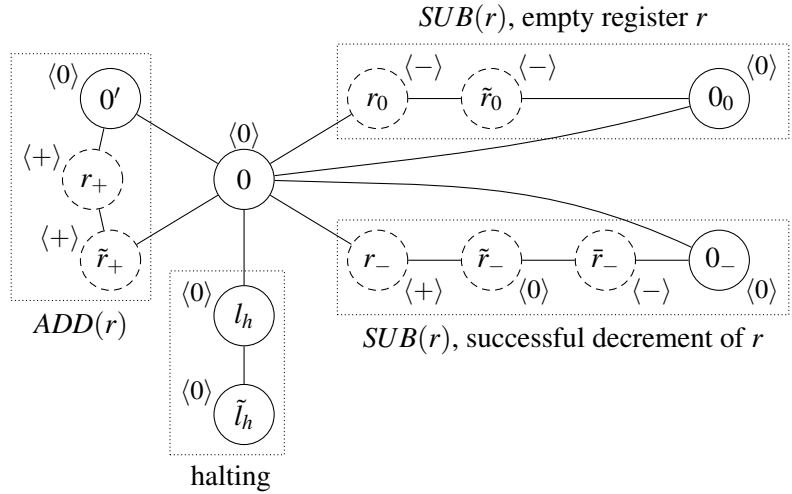

\begin{align*}
\Pi 	&= \left( L,V,T,R,(0,l_0),\tilde{l}_h,\pi _L,\pi _V,\varphi _s,G\right) ,\\
L	&=\{0,0^{\prime },0_0,0_-,l_{h},\tilde{l}_{h}\}
	   \cup \{ r_+,r_0,r_-, \tilde{r}_+,\tilde{r}_0, \tilde{r}_-, \hat{r}_- 
	   	\mid 1\leq r\leq m\},\\
V 	&=\{a_r,{a_r}^{-},{a_r}^{+}\mid 1\leq r\leq m \}
	     \cup \{p,p^+,p^-\mid p\in B\},\\
T 	&=\{a_r\mid 1\leq r\leq k \}.
\end{align*}

The evaluation $\pi_V$ for the symbols in $V$ corresponds to the superscript 
of the symbol, i.e., for $\alpha^z\in V$ with $z\in \{+,0,-\}$ we define 
$\pi_V (\alpha^0)=0$ (we usually omit the superscript $0$),  
$\pi_V (\alpha^+)=+1$, and $\pi_V (\alpha^-)=-1$.

The connection structure, i.e., the undirected graph $G$, as well as the 
polarizations of the cells given by $\pi _L$ can directly be derived from 
the graph depicted in Figure 2. The rules from $R$ are grouped in five
different groups; $R$ is the union of all the sets $R_u$, $u\in L$ as 
defined below:

\begin{description}
\item[root cell $0$ ] All simulations start from cell $0$ and again end there.
\newline
$R_{0}=\{ p\to p\mid p:\left( ADD\left( r\right) ,q,s\right) \in P\}
	\cup \{ p\to p^+,p\to p^-\mid p:\left( SUB\left( r\right) ,q,s\right) \in P\}
	\cup \{ l_{h}\to l_{h}\}$
	
\item[increment group ] Any ADD-instruction 
$p:\left( ADD\left( r\right) ,q,s\right)$ is simulated by passing from cell $0$ 
to $0^{\prime}$, from where only the correct path through $r_+$ and then
$\tilde{r}_+$ for the suitable $r$ will lead back to cell $0$.\newline
$R_{0^{\prime}}=\{ p\to p^+\mid p:\left( ADD\left( r\right) ,q,s\right) \in P\}$

$R_{r^{+}}=\{ \lambda \to a_r\}$ In order to guarantee that the rule 
$\lambda \to a_r$ is applied only once, we need the condition that 
after the application of a rule the vesicle has to leave the cell, which 
here means to pass to cell $\tilde{r}^{+}$ where the polarization is 
changed so that the vesicle will not be able to immediately return 
to cell $r^{+}$.

$R_{\tilde{r}^{+}}=\{ p^+\to q,p^+\to s
	\mid p:\left( ADD\left( r\right) ,q,s\right) \in P\}$

We observe that no vesicle with a $p^+$ can go from cell $0$ to
cell $\tilde{r}^{+}$ without the vesicle then immediately being caught there 
in cells $\tilde{r}^{+}$ and $r^{+}$, as the $p^+$ from cell $0$ is for 
a SUB-instruction and the rules in $\tilde{r}^{+}$ are for labels of 
ADD-instructions.

\item[zero check ] $R_{r_0}=\{a_r\to {a_r}^+ \}$.
Cell $0$ sends the vesicle to $r_0$ by non-deterministically applying 
the rule $p\to p^-$ and thus setting the polarization of the multiset to $-1$.
If the rule $a_r\to {a_r}^+$ is applicable, then the polarization goes back 
to $0$ and therefore the correct continuation in cell $\tilde{r}_0$ is 
blocked. On the other hand, when the vesicle returns back to cell $0$,
no rule can be applied there, and then moving to cell $0^{\prime}$ or 
cell $l_h$ also does not yield a successful continuation.

$R_{\tilde{r}_0}=\{ p^-\to s
	\mid p:\left( SUB\left( r\right) ,q,s\right) \in P\}$

Cell $0_0$ is needed for blocking the way from cell $0$ to cell $\tilde{r}_0$.

The rule set $R_{0_0}$ can be taken to be empty. If for formal reasons
one would not like to have such a situation where a vesicle can pass through 
a cell without undergoing an evolution rule, we could take:

$R_{0_0}= \{ s\to s\mid s\in B\}$

\item[decrement ] Passing the sequence of cells 
$0$--$r_-$--$\tilde{r}_-$--$\bar{r}_-$--$0_-$ allows for decrementing the 
number of symbols $a_r$. Cell $0$ sends the vesicle to $r_-$ by 
non-deterministically applying the rule $p\to p^+$ and by setting the 
polarization of the multiset to $+1$.

$R_{r_-}=\{ a_r\to {a_r}^-\}$ After the application of the rule $a_r\to {a_r}^-$ 
the polarization is again $0$, so the vesicle might also go back to 
cell $0$, but all possible continuations from there finally get blocked with 
the $p^+$ in there for a label $p$ of a SUB-instruction when moving into the 
increment group.

$R_{\tilde{r}_-}=\{ p^+\to q \}$ With the application of the rule $p^+\to q $ 
the polarization changes; if the wrong $r$-branch has been chosen from 
cell $0$, the computation gets stuck here.

$R_{\bar{r}_-}=\{ {a_r}^-\to \lambda \}$

As in the zero-check group, the set $R_{0_0}$ can be chosen to be 
empty or we take:

$R_{0_0}= \{ s\to s\mid s\in B\}$


\item[halting group ] As soon as $M$ has reached the $HALT$-label $l_h$,
we may pass to cell $l_h$ containing the rule $l_h \to \lambda $; 
the resulting vesicle then can go to the output cell $\tilde{l}_h$ to 
yield the result of the computation.
\end{description}

In the way described above $\Pi $ can simulate the computations of $M$.
If the vesicle reaches the output cell $\tilde{l}_h$, only terminal symbols 
from $\{a_r\mid 1\leq r\leq k \}$ are contained in its multiset which 
represents the $k$-dimensional vector computed by $M$ by the number 
of symbols $a_r$ for the number contained in register~$r$. 

\hfill
\end{proof}


\section{Conclusion and future research}

In this paper, we have investigated tissue P systems operating on vesicles of 
multisets with point mutations, i.e., with insertion, deletion, and substitution of 
single symbols, working either in the maximally parallel set derivation mode or 
in the sequential derivation mode. Without any additional control features, 
when using the sequential derivation mode, we obtain a characterization of the 
sets of (vectors of) natural numbers generated by partially blind register machines, 
whereas when using all three operations insertion, deletion, and substitution on 
the vesicles of multisets we can generate every recursively enumerable set of 
(vectors of) natural numbers. If we add the feature of elementary polarizations 
$-1,0,1$ to the multisets and to the cells of the tissue P systems, even sequential 
tissue P systems are computationally complete.
\medskip

Besides the maximally parallel set derivation mode, also the other set 
derivation modes (see~\cite{Alhazovetal2016set}) promise to yield similar 
results. Another topic is to investigate the influence of the underlying 
communication structure on the generative power, especially in the case of 
polarized tissue P systems. Moreover, complexity issues like the number 
of cells remain to be investigated in the future, for example, also with respect 
to find small universal devices, e.g., see~\cite{Alhazovetal2016set}. 
We may also consider tissue P systems with more than one vesicle 
moving around, which, for example, offers the possibility to require the whole 
system to halt in order to obtain a result. Finally, using different evaluation 
functions may have an influence on the descriptional complexity of 
polarized tissue P systems.


\bibliographystyle{eptcs}
\bibliography{AFL2017AFIV-final}


\end{document}